\documentclass[11pt,a4paper]{article}
\usepackage{dsfont}
\usepackage[sort]{cite}
\usepackage{ccfonts,eulervm} \usepackage[T1]{fontenc}

\usepackage{thmtools}
\usepackage{thm-restate}
\usepackage[pdfencoding=auto]{hyperref}
\usepackage{bookmark}
\usepackage{amsmath,amssymb,amsthm}
\usepackage[margin=1in]{geometry}
	\usepackage{hyperref}
\hypersetup{colorlinks=true,linkcolor=black,citecolor=[rgb]{0.5,0,0}}
\usepackage{csquotes}
\usepackage{tablefootnote}
\usepackage{multirow}
\usepackage{nicefrac}
\usepackage{graphicx}
\usepackage{tabularx}
\usepackage{makecell}
\renewcommand{\P}{\mathcal{P}}
\usepackage{diagbox}
\usepackage{xspace}

\usepackage{setspace}
\theoremstyle{plain}
\newtheorem{theorem}{Theorem}[section]

\newtheorem{corollary}[theorem]{Corollary}

\newtheorem{hypothesis}[theorem]{Hypothesis}

\renewcommand{\tilde}{\widetilde}
\renewcommand{\hat}{\widehat}

\newcommand{\N}{\mathbb N}

\newcommand{\SAT}{$\mathsf{SAT}$\xspace}

\newcommand\ETH{\texorpdfstring{$\mathsf{ETH}$}{ETH}\xspace}

\newcommand\poly{\text{poly}}
\newcommand{\cnf}{$\mathsf{CNF}$\xspace}

\newcommand{\CSP}{$\mathsf{CSP}$\xspace}

\title{Conditional lower bounds for sparse parameterized 2-CSP:\\ A streamlined proof}
\date{}
\author{%
Karthik C. S.\thanks{Rutgers University, USA, \texttt{karthik.cs@rutgers.edu}} \and
D\'{a}niel Marx\thanks{CISPA, Germany, \texttt{marx@cispa.de}} \and
Marcin Pilipczuk\thanks{Institute of Informatics, University of Warsaw, Poland, \texttt{m.pilipczuk@mimuw.edu.pl}} \and
U\'{e}verton Souza\thanks{Universidade Federal Fluminense (UFF), Brazil, \texttt{ueverton@ic.uff.br}}}
\begin{document}

\maketitle

\begin{abstract}
    Assuming the Exponential Time Hypothesis (ETH), a result of Marx (ToC'10) implies that there is no $f(k)\cdot n^{o(k/\log k)}$ time algorithm that can solve 2-CSPs with $k$ constraints (over a domain of arbitrary large size  $n$) for any computable function $f$. This lower bound is widely used to show that certain parameterized problems cannot be solved in time $f(k)\cdot n^{o(k/\log k)}$ time (assuming the ETH). The purpose of this note is to give a streamlined proof of this result. 
\end{abstract}

\section{Introduction}
The goal of this note is to discuss a simple proof for a widely used conditional lower bound on the time needed to solve Constraint Satisfaction Problems (\CSP). Our focus is on 2-CSPs, that is, the special case where each constraint involves two variables. Formally,
an instance  $\Gamma$ of 2-\CSP  consists of a (constraint) graph $H$,   an alphabet set $\Sigma$, and, for each edge $\{u, v\} \in E(H)$, a constraint $C_{uv} \subseteq \Sigma \times \Sigma $. An \emph{assignment} for $\Gamma$  is a function $\sigma:V(H)\to\Sigma$. An edge $\{u, v\} \in E(H)$ is said to be \emph{satisfied} by an assignment $\sigma$ if and only if $(\sigma(u), \sigma(v)) \in C_{uv}$. The goal of the 2-\CSP is to determine if there is an assignment $\sigma$ that satisfies all the edges (i.e., constraints). Such an assignment is referred to as a satisfying assignment. 
A 2-\CSP admitting a satisfying assignment is said to be satisfiable.  Some authors use the name {\em binary \CSP} for 2-\CSP, emphasizing that each constraint is defined by a binary relation on two variables. The problem can be equivalently stated as \textsc{Partitioned Subgraph Isomorphism}, but here we prefer to present the results using the language of 2-\CSP.

Let us consider two simple examples how a 2-\CSP can model other algorithmic problems. If $\Sigma=[3]$ and every constraint $C_{uv}$ is the inequality relation (i.e., $([3]\times [3])\setminus \{(1,1),(2,2),(3,3)\}$), then the problem is the same as \textsc{3-coloring}: the instance is satisfiable if and only if the constraint graph has a proper 3-coloring of the vertices. The \textsc{$k$-Clique} problem on a graph $G$ can be reduced to 2-\CSP in the following way. Let the alphabet set be $\Sigma=V(G)$ and let the constraint graph be a clique
on $k$ variables $v_1$, $\dots$, $v_k$. For every $1\le i < j \le k$, the constraint $C_{v_iv_j}$ is the relation $\{(x,y)\in V(G)\times V(G) \mid \text{$x$ and $y$ are adjacent in $G$}\}$. Then it is easy to see that the 2-\CSP instance is satisfiable if and only if $G$ 
has a clique of size $k$. Note that in the second example, the alphabet size can be arbitrarily large. The lower bounds we consider are mostly relevant for such instances.

Given a 2-\CSP instance with $k$ variables, an exhaustive search algorithm can solve it by considering each of the $|\Sigma|^{k}$ possible assignments. In a sense, the reduction from \textsc{$k$-Clique} shows that this algorithm is essentially optimal. Chen et al.~\cite{CHKX06a,CHKX06b} showed that there is no $f(k)\cdot n^{o(k)}$ time algorithm for \textsc{$k$-Clique} for any computable function $f$, unless the Exponential-Time Hypothesis (ETH) fails (see Section~\ref{sec:prelim} for the formal definition of ETH). The reduction given above translates this result to a lower bound for 2-\CSP in a transparent way.
\begin{theorem}[Chen et al.~\cite{CHKX06a,CHKX06b}]\label{thm:denselower}
If there is an $f(k)\cdot |\Sigma|^{o(k)}$ time algorithm for 2-\CSP, where $\Sigma$ is the alphabet set, $k$ is the number of variables, and $f$ is a computable function, then the ETH fails. 
\end{theorem}
Theorem~\ref{thm:denselower} is essentially tight (up to a constant factor in the exponent of $|\Sigma|$), but let us observe that the reduction from \textsc{$k$-Clique} creates dense 2-\CSP instances having $k$ variables and $\binom{k}{2}$ constraints. The lower bound does not remain valid if we consider sparse instances. Statements about sparse instances can be formulated by expressing the running time as a function of the number $m$ of constraints. As the reduction above creates an instance with $m=\Omega(k^2)$ constraints, the exponent cannot be better than $O(\sqrt{m})$.
\begin{corollary}
    \label{thm:sparseweak1}
If there is an $f(m)\cdot |\Sigma|^{o(\sqrt{m})}$ time algorithm for 2-\CSP, where $\Sigma$ is the alphabet set, $m$ is the number of constraints, and $f$ is a computable function, then the ETH fails. 
\end{corollary}
A different way to state the results about sparse instances is to consider only instances where the constraint graph has bounded degree. There is a very simple transformation that turns a 2-\CSP instance with $k$ variables and $m$ constraints into a 2-\CSP instance with $2m$ variables and a 3-regular  constraint graph: if a variable appears in $c$ constraints, then let us replace it with $c$ copies, connect these copies with a cycle of equality constraints, and let us introduce the original constraints in such a way that each copy of a variable appears in exactly one such constraint. A further simple transformation can be used to make the constraint graph bipartite, at the cost of increasing the number of variables by a constant factor. Corollary~\ref{thm:sparseweak1} and these transformations give the following slightly stronger statement.
\begin{corollary}\label{thm:sparseweak2}
If there is an $f(k)\cdot |\Sigma|^{o(\sqrt{k})}$ time algorithm for 2-\CSP on bipartite 3-regular constraint graphs, where $\Sigma$ is the alphabet set, $k$ is the number of constraints, and $f$ is a computable function, then the ETH fails. 
\end{corollary} 

Marx~\cite{marx-toc-treewidth} improved the lower bound to a result that is tight up to a logarithmic factor in the exponent.
\begin{theorem}[Marx~\cite{marx-toc-treewidth}]\label{thm:sparsemarx}
If there is an $f(k)\cdot |\Sigma|^{o(k/\log k)}$ time algorithm for 2-\CSP on bipartite 3-regular constraint graphs, where $\Sigma$ is the alphabet set, $k$ is the number of constraints, and $f$ is a computable function, then the ETH fails. 
\end{theorem}
Theorem~\ref{thm:sparsemarx} was obtained as a special case of a result in a much more general setting. Informally, Marx~\cite{marx-toc-treewidth} considered the 2-\CSP problem restricted to an arbitrary class $\mathcal{H}$ of constraint graphs and showed that, under any such restriction, an algorithm with running time $f(k)\cdot |\Sigma|^{o(t/\log t)}$ would violate the ETH, where $k$ is the number of variables, $t$ is the treewidth\footnote{Informally, treewidth is a measure of the proximity of the graph to a tree, and more formally is the size of the largest vertex set in a tree decomposition of the graph.} of the constraint graph, and $f$ is an arbitrary computable function. Specializing this result to a class of 3-regular bipartite graphs where treewidth is linear in the number of vertices (explicit constructions of bounded-degree expanders can be used to construct such classes) yields Theorem~\ref{thm:sparsemarx}.

\paragraph{Applications.}
An important aspect of Theorem~\ref{thm:denselower} is that it
can be used to obtain lower bounds for other parameterized problems.
W[1]-hardness proofs are typically done by parameterized reductions from \textsc{$k$-Clique}.
It is easy to observe that if a parameterized reduction increases the parameter only at most by a constant multiplicative factor, then this implies a lower bound
similar to Theorem~\ref{thm:denselower} for the target problem. In many cases, the reduction from \textsc{$k$-Clique} constructs an instance consisting of $k$ vertex-selection gadgets such that each gadget increases the parameter by a constant, and thus the parameter of the target instance is indeed linear in $k$. However, many of the more involved reductions use edge
selection gadgets (see e.g., \cite{DBLP:journals/tcs/FellowsHRV09}).
As the $k$-clique has $\Theta(k^2)$ edges, this means that the
reduction increases the parameter to $\Theta(k^2)$ and we can conclude
only the weaker bound stating that there is no $f(k)\cdot n^{o(\sqrt{k})}$-time algorithm for the
target problem (unless the ETH fails). If we want to obtain stronger
bounds on the exponent, then we have to avoid the quadratic blow-up of
the parameter and do the reduction from a different problem. One
possibility is to reduce from 2-\CSP on a 3-regular constraint graph and use the lower bound of Theorem~\ref{thm:sparsemarx}.  In such a reduction, we would need $k$ vertex-selection gadgets and $3k/2$ edge-selection gadgets. If each gadget increases the target parameter only by a constant, then the target parameter is again $O(k)$. Therefore, such a reduction and Theorem~\ref{thm:sparsemarx} allows us
to conclude that there is no $f(k)\cdot n^{o(k/\log k)}$-time
algorithm for the target problem.
The use of Theorem~\ref{thm:sparsemarx} has become a standard technique when proving almost-tight lower bounds ruling out $f(k)\cdot n^{o(k/\log k)}$ time algorithms
\cite{DBLP:conf/esa/MarxP15,DBLP:journals/jcss/JansenKMS13,DBLP:conf/stoc/CurticapeanDM17,DBLP:journals/siamdm/JonesLRSS17,DBLP:conf/focs/CurticapeanX15,DBLP:conf/esa/BonnetM16,DBLP:journals/algorithmica/GuoHNS13,DBLP:journals/toct/PilipczukW18a,DBLP:journals/dam/BonnetS17,DBLP:conf/esa/Bringmann0MN16,DBLP:journals/corr/abs-1808-02162,DBLP:conf/soda/LokshtanovR0Z20,DBLP:journals/corr/abs-1802-08189,DBLP:conf/iwpec/BonnetGL17,DBLP:journals/jacm/Cohen-AddadVMM21,DBLP:journals/siamcomp/ChitnisFHM20,DBLP:journals/talg/ChitnisFM21,DBLP:journals/algorithmica/AgrawalPSZ21,DBLP:conf/esa/CurticapeanDH21,DBLP:conf/ijcai/DeligkasEG22,DBLP:conf/aaai/KnopSS22,DBLP:journals/tcs/EibenGNRWY23,DBLP:journals/ipl/AmiriKMR19,DBLP:journals/algorithmica/GuoHNS13,DBLP:journals/mp/FominPRS23,DBLP:conf/aaim/CramptonCGJR13,DBLP:conf/esa/BonnetIJK19,DBLP:journals/algorithmica/BonnetCMP20,DBLP:journals/algorithmica/NederlofS22}. It seems that some kind of edge representation   is required in many W[1]-hardness proofs, and for such problems basing the reduction on Theorem~\ref{thm:denselower} would be able to rule out only algorithms with running time $f(k)\cdot n^{o(\sqrt{k})}$. For these problems, Theorem~\ref{thm:sparsemarx} is the only know way of obtaining a lower bound (almost) matching the $n^{O(k)}$ upper bounds. 





\paragraph{Our contribution.}
Our main contribution is to provide a streamlined proof of Theorem~\ref{thm:sparsemarx} and state some stronger formulations. As mentioned above,  Theorem~\ref{thm:sparsemarx} was obtained by considering the special case of expander graphs in the more general result of \cite{marx-toc-treewidth}. It turns out that the proof can be significantly shortened in this case: a multi-step combinatorial proof can be replaced by a known result about routing in expanders with congestion $O(\log n)$ \cite{DBLP:journals/jacm/LeightonR99}.

Instead of saying that the exponent cannot be $o(k/\log k)$, we state the result in a stronger form by stating a lower bound for every fixed $k$. A similar formulation was stated by Cohen-Addad et al.~\cite{DBLP:journals/jacm/Cohen-AddadVMM21}. It seems that such a statement is particularly useful to obtain in a clear way lower bounds that involve multiple parameters (where the little-$o$ notation is unclear). 
\begin{restatable}{theorem}{main}\label{thm:main} 
 There exists $\alpha>0$ and $k_0\in\mathbb{N}$ such that the following holds. If there exists some fixed even integer $k\ge k_0$ and an algorithm $\mathcal{A}_k$ with the following guarantees, then \ETH is false.
 \begin{description}
 \item[Input:] $\mathcal{A}_k$ takes as input a 2-\CSP instance $\Gamma(H,\Sigma,\{C_{uv}\}_{\{u,v\}\in E(H)})$ where $H$ is a 3-regular simple bipartite graph on $k$ vertices. 
 \item[Output:] $\mathcal{A}_k$ outputs 1 if and only if there is a satisfying assignment to $\Gamma$.
     \item[Run time:] $\mathcal{A}_k$ runs in time $O(|\Sigma|^{\alpha k/\log k})$.
 \end{description} 
\end{restatable}

Clearly, Theorem~\ref{thm:main} implies Theorem~\ref{thm:sparsemarx}: for a fixed sufficiently large (even) value of $k$, the hypothetical algorithm of Theorem~\ref{thm:sparsemarx} would satisfy the guarantees in Theorem~\ref{thm:main}.

Jaffke et al.~\cite{DBLP:conf/soda/JaffkeLMPS23}  proved a version of Theorem~\ref{thm:sparsemarx} where there is an upper bound on the alphabet size by a function of $k$; this result was used to establish a lower bound for the problem \textsc{Global Label MinCut}. Intuitively, if we set the alphabet size to $|\Sigma|=2^k$, then we expect that $|\Sigma|^{O(k)}=2^{O(k^2)}$ is the best possible running time. However, the formulations of Theorems~\ref{thm:sparsemarx} and \ref{thm:main} cannot be used to establish such results, as the number of possible instances with this restriction is bounded for each value of $k$. The following 
theorem proves bounds of this form; for example, it shows that, assuming ETH, there is no $2^{o(k^2/\log k)}$ time algorithm even with the restriction $|\Sigma|\le 2^k$.

\begin{restatable}{theorem}{mainalso}\label{thm:mainwithupperbound}
Let $f:\mathbb{N}\to\mathbb{N}$. We say $f$ is \emph{good} if it is non-decreasing, unbounded, and for all $n\in\mathbb{N}$, $f(n)$ can be computed in $\poly(f(n))$ time.    If there is a function $f$ that is good and an algorithm $\mathcal{A}$
with the following guarantees, then ETH is false.
\begin{description}
\item[Input:] $\mathcal{A}$ takes as input a 2-\CSP instance $\Gamma(H,\Sigma,\{C_{uv}\}_{\{u,v\}\in E(H)})$ where $H$ is a 3-regular simple bipartite graph on $k$ vertices and $|\Sigma| < f(k)$.
\item[Output:] $\mathcal{A}$ outputs 1 if and only if there is a satisfying assignment to $\Gamma$.
\item[Run time:] $\mathcal{A}$ runs in time at most
$f(k-1)^{o(k/\log k)}$.
\end{description}
\end{restatable}

We remark that under \#\ETH~\cite{DBLP:journals/talg/DellHMTW14}, the conditional time lower bounds in the above two theorems hold for the counting version of the 2-\CSP problem (even under parsimonious reductions). 


\paragraph{Original Proof in \cite{marx-toc-treewidth}.}
At a high level, the proof of Marx~\cite{marx-toc-treewidth} can be described in the following way. It is known that a $2^{o(n)}$ algorithm for 3-coloring an $n$-vertex 4-regular graph would violate the ETH. Our goal is to use the hypothetical $f(k)\cdot n^{o(t/\log t)}$ algorithm for some graph class $\mathcal{H}$ to efficiently solve the 3-coloring problem (given the reduction earlier, it will be convenient to consider 3-coloring as a 2-\CSP on $|\Sigma|=3$). Let us choose an $H\in \mathcal{H}$ with sufficiently large treewidth.
The core of the result is a combinatorial embedding result showing that a graph with $O(n)$ edges can be efficiently embedded in graphs with large treewidth. 

The notion of embedding we need is the following: it is a mapping $f$ from the vertices of $G$ to connected vertex sets of $H$ such that if two vertices $x$ and $y$ of $G$ are adjacent, then $f(x)$ and $f(y)$ intersect or have an edge between them (such an embedding is referred to as \emph{connected embedding} in this paper). The depth of the embedding is the maximum number of times each vertex of $H$ is used as an image of some vertex of $G$. The main combinatorial result of \cite{marx-toc-treewidth} shows that a bounded-degree graph $G$ with $n$ vertices has an embedding into $H$ of depth $O(n\log t/ t)$, where $t$ is the treewidth of $H$ (and this embedding can be found efficiently). The embedding is constructed via the following steps: (1) large treewidth implies the existence of a large set without balanced separators\footnote{Given a nonempty set $W$ of vertices, we say that a set $S$ of vertices is a balanced separator (with respect
to $W$) if $|W \cap C| \le |W|/2$ for every connected component $C$ of $G\setminus S$.}, (2) a set without balanced separators implies the existence of a large uniform concurrent flow, (3) the paths appearing in this large concurrent flow can be used to embed a blowup of the line graph of a complete graph, (4) the bounded-degree graph $G$ can be embedded into  this blowup of the line graph of a clique by a simple routing scheme.

Given such an embedding, we can reduce a 2-\CSP on $|\Sigma|=3$ and having constraint graph $G$ to a 2-\CSP on constraint graph $H$, but having larger alphabet $\Sigma'$. If the embedding maps vertex $v$ of $G$ to a subset of vertices $f(v)$ of $H$, then every variable in $f(v)$ ``simulates'' variable $v$ in the constructed instance. If the depth of the embedding is at most $d$, then each variable of the constructed instance needs to simulate at most $d$ original variables. Therefore, if the alphabet size of the new instance is $|\Sigma'|=|\Sigma|^d=3^d$, then each variable of $H$ can simulate the required number of original variables. As $f(v)$ is connected in $H$, we can introduce constraints that ensure that every variable in $f(v)$ expresses the same value for $v$. Furthermore, when $u$ and $v$ are adjacent in $G$, then $f(u)$ and $f(v)$ intersect or adjacent in $H$, so we can introduce a constraint to ensure that the  constraint $C_{uv}$ on $u$ and $v$ are respected. Thus the created instance is equivalent to the original one. By carefully following how the instance parameters change during the reduction, we can argue that the hypothetical algorithm on $H$ would solve the original 3-coloring instance in $2^{o(n)}$ time.

One technical difficulty in the proof is choosing $H\in \mathcal{H}$. To obtain $2^{o(n)}$ running time for \textsc{3-Coloring}, the choice of $H$ should depends on $n$: for larger $n$, we want to find an $H$ with larger treewidth in order to have a larger advantage compared to brute force. Thus one has to take into account the time required to enumerate graphs from $\mathcal{H}$. On the other hand, graph $H$ cannot be very large, otherwise the factor $f(k)$ in the running time could be too large. Thus the correct choice of $H$ creates an additional layer of difficulties in the proof.

\paragraph{Our Proof.} Our main observation is that if the goal is to obtain Theorem~\ref{thm:sparsemarx}, then it can be reached in a much simpler way than going through a general bound for treewidth and restricting it to the special case of expanders. In particular, the required embedding result is much simpler if the target graph is an expander. The key step is provided by Leighton and Rao~\cite[Theorem 22]{DBLP:journals/jacm/LeightonR99}, showing the following routing result: given disjoint pairs of vertices $(x_1,y_1)$, $\dots$, $(x_t,y_t)$ in a $k$-vertex expander graph, we can find in polynomial time $t$ paths $P_1$, $\dots$, $P_t$ such that $P_i$ has endpoints $(x_i,y_i)$ and every edge is used by $O(\log k)$ of the paths.

Given an $n$-vertex  bounded-degree graph $G$ and a $k$-vertex bounded-degree expander $H$, we can use this routing result to find an embedding of $G$ in $H$ having depth $O(n\log k/k)$. First, let $f$ be an arbitrary mapping from the vertices of $G$ to vertices of $H$ such that at most $n/k$ vertices are mapped to each vertex of $H$. This mapping can be used to define a ``demand multigraph'' $D$: it has the same vertex set as $H$, and for every edge $uv$ of $G$, we introduce a new edge $f(u)f(v)$ into $D$. As $G$ has bounded degree and $f$ maps at most $n/k$ vertices to each vertex of $H$, it follows that $D$ has maximum degree $O(n/k)$ and hence its edge set can be partitioned into $O(n/k)$ matchings. For each such matching, we can use the routing result of  Leighton and Rao~\cite[Theorem 22]{DBLP:journals/jacm/LeightonR99} to find paths connecting the endpoints.
This means that for every edge $uv$ of $G$, we find a path in $H$ connecting $f(u)$ and $f(v)$ in such a way that every edge of $H$ is used by $O(n\log k/k)$ of the paths. As $H$ has bounded degree, it also follows that every vertex is used by $O(n\log k/k)$ paths. Now we can define the embedding by letting $h(u)$ be the union of the paths corresponding to the edges of $G$ incident to $u$. It is clear that every vertex of $H$ is used by $O(n\log k/k)$ such sets and if $uv$ is an edge of $G$ then $h(u)$ and $h(v)$ intersect.


Another technical difference compared to the original proof of Marx~\cite{marx-toc-treewidth} is that the original proof started with the assumption that there is no $2^{o(n)}$ time algorithm for $n$-variable 3SAT. This is a somewhat weaker assumption than the ETH, which states that there is no $O(2^{\epsilon n})$ algorithm for some $\epsilon>0$. Using the ETH as the starting assumption and stating the result for a fixed $k$ as in Theorem~\ref{thm:main} allows a proof with significantly fewer technicalities.

\paragraph{Obstacles to obtain tight result}
It remains an important open question whether the $\log k$ factor in Theorem~\ref{thm:sparsemarx} can be reduced or even eliminated. One obvious approach for improvement would be to improve the embedding result and show that an $n$-vertex bounded-degree graph $G$ has an embedding with depth $o(n
\log k/k)$ into some $k$-vertex expander. Note that the proof of our embedding result starts with an arbitrary balanced mapping of vertices of $G$ to vertices of $H$, so there seems to be a lot of room for improvement for optimizing this scheme by a more careful grouping of vertices. However, by now, it is known that the $\log k$ factor in the embedding result cannot be improved in general.
\begin{theorem}[Alon and Marx~\cite{DBLP:journals/siamdm/AlonM11}]\label{th:alonmarx1}
Let $H$ be a 3-regular graph with $k$ vertices. Then, for all even
$n>n_0(k)$, there exists a 3-regular graph $G$ on $n$ vertices so that
any embedding  into $H$ is of depth at least $\Omega(\frac{n\log k}{k})$.
\end{theorem}

Thus any improvement of Theorem~\ref{thm:sparsemarx} should involve significantly different techniques, going beyond the simple notion of embedding we have here.

\subsection{Organization of Paper}
In Section~\ref{sec:prelim}
we introduce the problems and hypotheses of interest to this paper. In Section~\ref{sec:Embedding} we prove the main technical tool (a graph embedding theorem) needed to prove Theorems~\ref{thm:main}~and~\ref{thm:mainwithupperbound}. Finally, in Section~\ref{sec:main}, we prove Theorems~\ref{thm:main}~and~\ref{thm:mainwithupperbound}.

\section{Preliminaries}\label{sec:prelim}

In this section, we state a few definitions of relevance to the rest of the paper. 

\paragraph{\bf $3$-SAT.}
In the $3$-SAT problem, we are given a \cnf formula $\varphi$ over $n$ variables $x_1,\ldots ,x_n$, such that each clause contains at most $3$ literals. Our goal is to decide if there exist an assignment to $x_1,\ldots ,x_n$ which satisfies $\varphi$.

\begin{hypothesis}[Exponential Time Hypothesis (\ETH)~\cite{IP01,IPZ01,Tovey84}] \label{hyp:eth}
There exists $\epsilon > 0$ such that no algorithm can solve 3-\SAT on $n$ variables in time $O(2^{\epsilon n})$. Moreover, this holds even when restricted to formulae in which each variable appears in at most three clauses.
\end{hypothesis}

Note that the original version of the hypothesis from~\cite{IP01} does not enforce the requirement that each variable appears in at most three clauses. To arrive at the above formulation, we first apply the Sparsification Lemma of~\cite{IPZ01}, which implies that we can assume without loss of generality that the number of clauses $m$ is $O(n)$. We then apply Tovey's reduction~\cite{Tovey84} which produces a 3-\cnf instance with at most $3m + n = O(n)$ variables and every variable occurs in at most three clauses. This means that the bounded occurrence restriction is also without loss of generality.


\paragraph{\bf $2$-\CSP.}
An instance  $\Gamma$ of 2-\CSP  consists of a (constraint) graph $H$,   an alphabet set $\Sigma$, and, for each edge $\{u, v\} \in E(H)$, a constraint $C_{uv} \subseteq \Sigma \times \Sigma $. An \emph{assignment} for $\Gamma$  is a function $\sigma:V(H)\to\Sigma$. An edge $\{u, v\} \in E(H)$ is said to be \emph{satisfied} by an assignment $\sigma$ if and only if $(\sigma(u), \sigma(v)) \in C_{uv}$. The goal of the 2-\CSP is to determine if there is an assignment $\sigma$ that satisfies all the edges (i.e., constraints). Such an assignment is referred to as a satisfying assignment. 
A 2-\CSP admitting a satisfying assignment is said to be satisfiable.

\paragraph{Cheeger Constant.} Let $G$ be a graph. For any subset $S\subseteq V(G)$, the edge boundary or cut of $S$, denoted $\delta(S)$, is the  set of edges which have one endpoint in S and the other endpoint not in $S$, i.e., 
$$\forall S\subseteq V(G),\ \delta(S):=\{\{u,v\}\in E(G)\mid u\in S,v\notin S\}.$$

The Cheeger constant of $G$, denoted $\alpha(G)$, is defined by:
$$\alpha(G):=\underset{\substack{S\subseteq V(G)\\ |S|\le |V(G)|/2}}{\min}\left\{\frac{|\delta(S)|}{|S|}\right\}.$$

\paragraph{Expanders.} We will now discuss about the construction of a simple bipartite 3-regular expander $G$. Our starting is the following result of Alon \cite{alon2021explicit}.

\begin{theorem}[Alon \cite{alon2021explicit}]\label{thm:alonexpanders}
    There is an algorithm $\mathcal{A}$ such that the following holds for every degree $d\ge 3$, every $\varepsilon>0$ and all sufficiently large $n \ge 
n_0(d, \varepsilon)$, where $n\cdot d$ is even. $\mathcal{A}$ on input $n, d,$ and $\varepsilon$ ($n$ is provided in unary to $\mathcal{A}$) outputs  in polynomial time a $d$-regular simple\footnote{Recall that in a simple graph there are no self-loops and no multiple edges between a pair of vertices. } graph $G$ on $n$ vertices, such that the the absolute
value of every non-trivial eigenvalue of the  adjacency matrix of $G$ is at most  $2\cdot \sqrt{d-1}+\varepsilon$. 
\end{theorem}

The above is a strengthening of Theorem 1.3 in \cite{alon2021explicit} (as the expander is simple in the above statement), but this follows from the proof in \cite{alon2021explicit} (which relies on \cite{mohanty2020explicit}).
Alon communicated to us \cite{alon}
 that the above result can be used to prove the existence of 3-regular simple balanced bipartite expanders of every even order (greater than 4).
 
\begin{theorem}[Alon \cite{alon}]\label{thm:alon bipartite}
    There is a universal constant $\alpha_0>0$ and an algorithm $\mathcal{A}$ such that the following holds. $\mathcal{A}$ on input an even integer $n\ge 6$  ($n$ is provided in unary to $\mathcal{A}$) outputs  in polynomial time  a $3$-regular simple balanced bipartite graph $G$ on $n$ vertices such that $\alpha(G)\ge \alpha_0$. 
\end{theorem}
\begin{proof}
Let $\mathcal{A}'$ be the algorithm given by Theorem~\ref{thm:alonexpanders} and let $n_0\in\mathbb{N}$ be the integer guaranteed in that theorem statement for $d:=3$ and $\varepsilon:=2.85-2\sqrt{2}$. The $\alpha_0$ in this theorem statement is then set to $\min(1/n_0,0.015)$.

If the input $n$ to $\mathcal{A}$ is less than $\max(2n_0,12)$ then simply output the following graph\footnote{Any arbitrary connected balanced bipartite graph $G_n$ on $n$ vertices would suffice for us.} on $n$ vertices: $\forall i\in [n/2]$, and $\forall t\in \{-1,0,1\}$, we have $(i,((i+t)\mod (n/2))+\frac{n}{2}+1)\in E(G_n)$. Note that $G_n$
 is connected, and thus for all non-empty $S\subseteq V(G_n)$ such that $|S|\le n/2$, we have that there is an edge in $\delta(S)$. Thus, we have:
 $$\alpha(G_n)\ge \left( \underset{\substack{S\subseteq V(G_n)\\ |S|\le n/2}}{\min}\ \frac{1}{|S|}\right)\ge \frac{2}{n}\ge \frac{1}{n_0}.$$

    If the input $n$ to $\mathcal{A}$ is at least   $\max(2n_0,12)$ and $n$ is divisible by 4, then let $G_n^0$ be the output of $\mathcal{A}'$ on input $n/2, d=3, \varepsilon:=2.85-2\sqrt{2}$ (we can invoke $\mathcal{A}'$ on input $n/2$ as $n/2$ is even). Let $G_n$ be the double cover of $G_n^0$, i.e., if $A$ is the adjacency matrix of $G_n^0$ then the adjacency matrix of $G_n$ is $\begin{bmatrix}
\textbf{0} & A  \\
A & \textbf{0}  
\end{bmatrix}$. Note that $G_n^0$ is a 3-regular graph on $n/2$ vertices and $G_n$ is a 3-regular simple balanced bipartite graph on $n$ vertices. Moreover, if $-2.85\le \lambda_{\frac{n}{2}}^0\le \lambda_{\frac{n}{2}-1}^0\le \cdots \le \lambda_2^0 \le 2.85 <\lambda_1^0=3$ were the eigenvalues of $G_n^0$ then the set of eigenvalues of $G_n$ is $\{(-1)^a\cdot \lambda_j^0\mid j\in [n/2],\ a\in\{0,1\}\}$. Thus the second highest eigenvalue of $G_n$ is at most 2.85. From Cheeger inequality \cite{cheeger1970lower,alon1986eigenvalues,alon1985lambda1}, we have:
$$\alpha(G_n)\ge (3-2.85)/2=0.075.$$

If the input $n$ to $\mathcal{A}$ is at least  $\max(2n_0,12)$ and $n$ is not divisible by 4 (i.e., $n\cong 2 \mod 4$), then let $G_{n+2}$ be the output of $\mathcal{A}$ on input $n+2$.  Note that  $G_{n+2}$ is a 3-regular balanced bipartite simple graph on $n+2$ vertices  with $\alpha(G_{n+2})\ge  0.075$.

Let $\hat{e}:=\{\hat{u},\hat{v}\}$ be an   edge in $G_{n+2}^0$ in (one of) the odd  cycle (denoted by $C$) of minimum size   (recall that $G_{n+2}^0$ is not bipartite as it's smallest eigenvalue is strictly greater than $-3$).   Let $e:=(u,v)$ be (one of the two) corresponding edges of $\hat{e}$ in $G_{n+2}$. Let $v,v_1,$ and $v_2$ (resp.\ $u,u_1,$ and $u_2$) be the neighbors of $u$ (resp.\ $v$) in $G_{n+2}$. Moreover, we may assume that (by relabeling $v_1$ by $v_2$ if needed) $(u_i,v_i)\notin E(G_{n+2})$ for all $i\in\{1,2\}$. This is because if $(u_1,v_2)\in E(G_{n+2})$ and $(u_1,v_1)\in E(G_{n+2})$ then we have a cycle of length 4 in $C$ which contradicts the minimaility of size of $C$.  We remove $u$ and $v$ (and the edges incident  to these two vertices), and insert an edge between $u_1$ and $v_1$, and between $u_2$ and $v_2$, to obtain $G_n$, a 3-regular \emph{simple} balanced bipartite graph on $n$ vertices. 

Let $S\subseteq V(G_n)$ such that $|S|\le |V(G_n)|/2$ and $\alpha(G_n)=\frac{|\delta(S)|}{|S|}$. Let $T:=S\ \cap\ \{u_1,u_2,v_1,v_2\}$. 
 Suppose $S=T$, then we note that $\delta(T)>0$ in $G_n$,  as otherwise, we have that the set of vertices $\{u,v,u_1,v_1,u_2,v_2\}$ is disconnected from the rest of the graph in $G_{n+2}$, and this implies $\alpha(G_{n+2})=0$ (as $n\ge 12$). But if $\delta(T)>0$ then we have: 
 $$\alpha(G_n)=\frac{|\delta(S)|}{|S|}=\frac{|\delta(T)|}{|T|}\ge \frac{1}{4}.$$

On the other hand, suppose $S\neq T$, then let $\tilde{S}:=S\setminus T$. First note that $\delta(S)$ in $G_n$ is not empty as then we have that either $\delta(S)$ or $\delta(S\cup\{u,v\})$ is empty in $G_{n+2}$. Let $F_{\tilde S}\subseteq E(G_n)$ (resp.\ $F_T\subseteq E(G_n)$) be the set of edges with one endpoint in $\tilde S$ (resp.\ $T$) and another endpoint in $V(G_n)\setminus S$. Since there are no edges between the set $\tilde S$ and $\{u,v\}$ in $G_{n+2}$, we have that $F_{\tilde S}$ is also the set of edges with one endpoint in $\tilde S$ and another endpoint in $V(G_{n+2})\setminus S$.  On the other hand, let $\tilde {F}_{T}\subseteq E(G_{n+2})$  be the set of edges with one endpoint in $T$   and another endpoint in $V(G_{n+2})\setminus S$, then we know that $0\le |\tilde F_T|-|F_T|\le 4$. Also since $\delta(S)$ in $G_n$ is not empty, we have that $|F_{\tilde S}|+|F_{T}|\ge 1$.
Thus, we have: 
$$|\tilde F_T|+|F_{\tilde S}|\le 4+|F_T|+|F_{\tilde S}|\le 4\cdot(|F_T|+|F_{\tilde S}|)+|F_T|+|F_{\tilde S}|= 5(|F_{\tilde S}|+|F_T|).$$
This implies:
$$\alpha(G_n)=\frac{|\delta_n(S)|}{|S|}\ge \frac{|F_{\tilde S}|+|F_{T}|}{|S|}\ge \frac{|F_{\tilde S}|+|\tilde{F}_{T}|}{5|S|}\ge \frac{\alpha(G_{n+2})}{5}\ge 0.015.$$

Finally, we note that $\mathcal{A}$ runs in polynomial time because $\mathcal{A}'$ runs in polynomial time.
\end{proof}

\section{Proof of Main Embedding Theorem}\label{sec:Embedding}

In this section, we prove a graph embedding theorem which will be used to prove Theorems~\ref{thm:main}~and~\ref{thm:mainwithupperbound} in the next section. First we introduce a notion of graph embedding below. 

Two connected subgraphs $H_1$ and $H_2$ of a graph $H$ \emph{touch}
if they share a vertex or if there is an edge of $H$ with one endpoint in $V(H_1)$
and another endpoint in $V(H_2)$. 
A \emph{connected embedding} of a graph $G$ in a graph $H$ is a function $\Psi:V(G)\to \P(V(H))$
that maps every $v \in V(G)$ to a nonempty connected subgraph in $H$ 
such that for every edge $uv \in E(G)$, the subgraphs $\Psi(u)$ and $\Psi(v)$ touch. 
The \emph{depth} of an embedding $\Psi$, denoted $\Delta(\Psi)$,
is: 
$$\Delta(\Psi):=\max_{x \in V(H)} |\{v \in V(G)~|~x \in \Psi(v)\}|,$$ that is, the maximum 
number of subgraphs $\Psi(v)$ that meet in a single vertex of $H$. 
In literature, such embedding also appear in the context of \emph{congested minor models}
where depth is called \emph{congestion} or \emph{ply}. 

The main embedding result of this section is the following. 

\begin{theorem}\label{thm:embedding}
There exists a constant $Z\ge 1$  and an algorithm $\mathcal{A}$ that takes as input  a graph $G$  and an even integer $k\ge 6$, and outputs a  bipartite 3-regular simple graph $H$ with no isolated vertices  and a connected embedding $\Psi:V(G)\to\P(V(H))$ such that the following holds.
\begin{description}
    \item[Size:] $|V(H)|\le k$. 
    \item[Depth Guarantee:]  $\Delta(\Psi)\le  Z\cdot \left(1+\frac{|V(G)|+|E(G)|}{k}\right) \cdot \log k$.
    \item[Runtime:] $\mathcal{A}$ runs in time $(|V(G)|+|E(G)|)^{O(1)}$.
\end{description}
\end{theorem}

Before, we proceed with the proof of the above theorem, we first state a key result of Leighton and Rao \cite{DBLP:journals/jacm/LeightonR99} on expander routing that will be used in our proof. 

\begin{theorem}[Leighton and Rao \cite{DBLP:journals/jacm/LeightonR99}]\label{thm:congestion}
   There exists  an algorithm $\mathcal{A}$ that takes as input a $k$-node bounded degree graph $G$,  and a $k$-node bounded degree graph $H$
with Cheeger constant $\alpha:=\alpha(H)$, and a bijection $h:V(G)\to V(H)$, and outputs  in deterministic polynomial time,   a path $P_{uv}$ for each edge $uv$ of $G$ such that the following holds.
\begin{enumerate}
\item Path $P_{uv}$ is a path between $h(u)$ and $h(v)$ in $H$,
\item Each edge of $H$ is in at most $O(\alpha^{-1}\log k)$ of the paths,
\item Path $P_{uv}$ has length at most $O(\alpha^{-1}\log k)$.
\end{enumerate}
\end{theorem}

\begin{proof}[Proof of Theorem~\ref{thm:embedding}]
Let $\mathcal{A}'$ (resp.\ $\mathcal{A}''$) be the algorithm given by Theorem~\ref{thm:congestion} (resp.\ Theorem~\ref{thm:alon bipartite}). Given as input  a $d$-regular graph $G$   and an even integer $k\ge 6$, the algorithm $\mathcal{A}$ does the following.   

$\mathcal{A}$ first provides $k$ as input to $\mathcal{A}''$. Let $H$ be the balanced bipartite 3-regular expander $H$ on $k$ vertices with Cheeger constant bounded away from $\alpha_0$ (some positive universal constant) that is given as output.

$\mathcal{A}$ then fixes an arbitrary function $g:V(G)\to[k]$  such that for all  $j\in[k]$, we have $|\{v\in V(G)\mid g(v)=j\}|\le \lceil |V(G)|/k\rceil$, i.e.,  $g$ is balanced.

Consider the following multigraph (\emph{demand graph}) $D$ on $k$ vertices, where for every $\{u,v\}\in  E(G)$, if $g(u)\neq g(v)$, then there is a unique edge in $E(D)$ between $g(u)$ and $g(v)$. In $D$, there are at most $\frac{|V(G)|d}{k}$ edges incident on any vertex.  Thus from \cite{de1972decomposition} $E(D)$ admits a decomposition into matchings, where the size of each matching is at least $\frac{k}{2}-1$. For each such matching, say $\{(x_1,y_1),\ldots ,(x_{\ell},y_{\ell})\}$, where $\ell\ge \frac{k}{2}-1$, we construct  the degree 1 graph $\hat{G}$ given by the union of the $\ell$ edges of the matching.

$\mathcal{A}'$ on input $\hat{G}$ and $H$ (and some canonical bijection between $V(\hat{G})$ and $V(H)$) 
outputs $\ell$ paths in $H$ whose congestion is at most $O(\log k)$ (since the Cheeger constant of $H$ is bounded below by a positive constant). We do the above for each matching (a different $\hat{G}$ for each matching) and using the same graph $H$ as input to $\mathcal{A}'$ to obtain a collection of paths. More importantly, there is a 1-to-1 correspondence between the edges of $G$ and this collection of paths. 

The final connected embedding $\Psi$ is given as follows: for each $v\in V(G)$, we have  $d$ many edges incident on it and each edge corresponds to a path of length at most $O(\log k)$. Then $\Psi(v)$ is the union of the nodes in all the $d$ paths corresponding to the $d$ edges incident of $v$. 

The runtime claim is straightforward. To see the claim on the depth of the embedding, notice that the congestion on each node of $H$ in each iteration of the calling of $\mathcal{A}'$ is at most $O(\log k)$. We make $O(|V(G)|d/k)$ many calls and thus the total congestion is $O(\frac{|V(G)|d\log k}{k})$.

Finally, if $|V_1|\neq |V_2|$, the only difference, would be in the selection of $g_1$ and $g_2$, where instead of the range being $[k/2]$ for both, it would instead need to be $k|V_1|/|V|$ and $k|V_2|/|V|$  respectively.
\end{proof}


\section{ETH Lower Bound for 2-CSP}\label{sec:main}

In  this section we first prove a bridging theorem to connect the embedding result from previous section with 2-\CSP{}s, and then in the subsequent subsection prove Theorems~\ref{thm:main}~and~\ref{thm:mainwithupperbound}.

\subsection{Embedding 2-\CSP in to a smaller constraint graph}

The below theorem is essentially proved in \cite{marx-toc-treewidth} but we provide a proof here for the sake of completeness. 

\begin{theorem}\label{thm:embedCSP}
\begin{sloppypar}There is an algorithm $\mathcal{A}$ that takes as input  a 2-\CSP instance $\Gamma(G(V,E),\Sigma,\{C_{uv}\}_{\{u,v\}\in E})$, a simple graph $H(X,F)$ without isolated vertices, and a connected embedding $\Psi: V\to \P(X)$ of depth $d$, and outputs a 2-\CSP instance $\Phi(H(X,F),\Sigma^d,\{\tilde{C}_{xy}\}_{\{x,y\}\in F})$ such that the following holds.\end{sloppypar}
\begin{description}
    \item[Reduction Guarantee:]  $\Gamma$ is satisfiable if and only if $\Phi$ is satisfiable.
    \item[Runtime:] $\mathcal{A}$ runs in time $O(|E|\cdot |F|\cdot |\Sigma|^{2d+2})$.
\end{description}
\end{theorem}
\begin{proof}
Given a 2-\CSP instance $\Gamma(G(V,E),\Sigma,\{C_{uv}\}_{\{u,v\}\in E})$, a graph $H(X,F)$, and a connected embedding $\Psi: V\to \P(X)$ of depth $d$, the algorithm $\mathcal{A}$ outputs the following 2-\CSP instance $\Phi(H(X,F),\Sigma^d,\{\tilde{C}_{xy}\}_{\{x,y\}\in F})$, where we only need to specify $\{\tilde{C}_{xy}\}_{\{x,y\}\in F}$. 

For every $x\in X$ define $V_x\subseteq V$ as follows: for all $v\in V$ we have
$v\in V_x$ if and only if $x\in \Psi(v)$. We further define $E_x$ to be all the edges in $E$ whose both end points are in $V_x$. Finally, define $d_x:=|V_x|$ and $\pi_x:V_x\to [d_x]$  to be some canonical bijection.  Note that since $\Psi$ has depth $d$ we know that $d_x\le d$ for all $x\in X$. 

Fix some $\{x,y\}\in F$. Let $V_{xy}\subseteq V$ and $E_{xy}\subseteq E$ be defined as follows.
$$V_{xy}:=V_x\cap V_y,\ \ E_{xy}:= \{\{u,v\}\in E\mid u\in V_x, v\in V_y\}.$$

Let $(a,b)\in \Sigma^d\times \Sigma^d$. We include $(a,b)\in \tilde{C}_{xy}$ if and only if both the following hold:

\begin{description}
    \item[Vertex Consistency:] For every $v\in V_{xy}$ we have $a_{\pi_x(v)}=b_{\pi_y(v)}$. 
    \item[Edge Consistency:] For every $\{u,v\}\in E_{xy}$ (resp.\ $\{u,u'\}\in E_{x}$ or $\{v,v'\}\in E_{y}$) we have $(a_{\pi_x(u)},b_{\pi_y(v)})\in C_{uv}$ (resp.\ $(a_{\pi_x(u)},a_{\pi_x(u')})\in C_{uu'}$ or $(b_{\pi_y(v)},b_{\pi_y(v')})\in C_{vv'}$). 
\end{description}

\paragraph{Completeness Analysis:} Suppose $\sigma:V\to \Sigma$ is a satisfying assignment to $\Gamma$ then consider the assignment $\tilde\sigma:X\to \Sigma^d$ to $\Phi$ constructed from $\sigma$ as follows. Let $\alpha\in \Sigma$ be an arbitrary element. For all  $x\in X$ we have:
$$
\tilde\sigma(x)_i:=\begin{cases}
\sigma(\pi_x^{-1}(i))&\text{ if }i\le d_x\\
\alpha&\text{ if }i> d_x
\end{cases}.
$$ 

We claim that $\tilde\sigma$ is a satisfying assignment to $\Phi$. To see this, consider an arbitrary $\{x,y\}\in F$. We verify that $(\tilde\sigma(x),\tilde\sigma(y))\in \tilde{C}_{xy}$ by checking the conditions of vertex consistency and edge consistency are met.  

For every $v\in V_{xy}$ we have by definition that $\tilde{\sigma}(x)_{\pi_x(v)}=\sigma(\pi_x^{-1}(\pi_x(v)))=\sigma(v)=\sigma(\pi_y^{-1}(\pi_y(v)))=\tilde{\sigma}(y)_{\pi_y(v)}$. This verifies vertex consistency. Next we verify edge consistency.

\begin{sloppypar}For every $\{u,v\}\in E_{xy}$ (resp.\ $\{u,u'\}\in E_{x}$ or $\{v,v'\}\in E_{y}$) we have $(\tilde{\sigma}(x)_{\pi_x(u)},\tilde\sigma(y)_{\pi_y(v)})=(\sigma(u),\sigma(v))$ (resp.\ $(\tilde{\sigma}(x)_{\pi_x(u)},\tilde\sigma(x)_{\pi_x(u')})=(\sigma(u),\sigma(u'))$ or $(\tilde{\sigma}(y)_{\pi_y(v)},\tilde\sigma(y)_{\pi_y(v')})=(\sigma(v),\sigma(v'))$) and this is  in $C_{uv}$ (resp.\ $C_{uu'}$ or $C_{vv'}$) because $\sigma$ is a satisfying assignment to $\Gamma$. \end{sloppypar}

\paragraph{Soundness Analysis:} Suppose $\tilde\sigma:X\to \Sigma^d$ is a satisfying assignment to $\Phi$ then consider the assignment $\sigma:V\to \Sigma$ to $\Gamma$ constructed from $\tilde\sigma$ as follows. 

Fix any function $\tau:V\to X$    that satisfies the following: for all $v\in V$, $\tau(v)\in\Psi(v)$. Informally, $\tau$ is picking a representative vertex in $\Psi(v)$ for each $v\in V$. For all  $v\in V$ we have:
$$
\sigma(v):=\tilde\sigma(\tau(v))_{\pi_{\tau(v)}(v)}.
$$ 

We first observe that the value $\sigma(v)$ is independent of the choice of $\tau(v)$, that is, 
\begin{equation}\label{eq:sigma-independent}
\forall {v \in V},\ \forall {x \in \Psi(v)},\quad \tilde\sigma(x)_{\pi_x(v)} = \sigma(v).
\end{equation}
To see~\eqref{eq:sigma-independent}, fix $v \in V$ and $x \in \Psi(v)$. 
Since $\Psi(v)$ is a connected subgraph of $H$, we have that there is a path from 
$\tau(v)$ to $x$ contained in $\Psi(v)$, and let $\tau(v) =: x_1-x_2-\cdots-x_r := x$ be one such path where for all
$i \in [r]$ we have $x_i \in \Psi(v)$. Since $\tilde\sigma$ is a satisfying assignment, for every $i \in [r-1]$,
by looking at the vertex consistency condition, we have $\tilde\sigma(x_i)_{\pi_{x_i}(v)} = \tilde\sigma(x_{i+1})_{\pi_{x_{i+1}}(v)}$. Thus, we have 
$\sigma(v) = \tilde\sigma(\tau(v))_{\pi_{\tau(v)}(v)} = \tilde\sigma(x_1)_{\pi_{x_1}(v)} = \tilde\sigma(x_r)_{\pi_{x_r}(v)} = \tilde\sigma(x)_{\pi_x(v)}$. This proves~\eqref{eq:sigma-independent}.

We claim that $\sigma$ is a satisfying assignment to $\Gamma$. To see this, consider an arbitrary $\{u,v\}\in E$. We verify below that $(\sigma(u),\sigma(v))\in {C}_{uv}$. 

From the definition of $\Psi$, we have that either (i) there exists $x\in X$ such that $x\in \Psi(u)\cap \Psi(v)$, or, (ii) there  exists $\{x,y\}\in F$ such that $x\in\Psi(u)$ and $y\in \Psi(v)$. 

\begin{description}
\item[Case (i):] By~\eqref{eq:sigma-independent},
we have $\sigma(u) = \tilde\sigma(x)_{\pi_x(u)}$ and $\sigma(v) = \tilde\sigma(x)_{\pi_x(v)}$.
Since $\{u,v\}\in E_x$ and $H(X,F)$ contains at least one edge incident with $x$, by the edge consistency condition 
on any such edge we have $(\sigma(u) = \tilde{\sigma}(x)_{\pi_{x}(u)},\sigma(v) = \tilde{\sigma}(x)_{\pi_{x}(v)})\in C_{uv}$, as desired.
\item[Case (ii):] By~\eqref{eq:sigma-independent},
we have $\sigma(u) = \tilde\sigma(x)_{\pi_x(u)}$ and $\sigma(v) = \tilde\sigma(y)_{\pi_y(v)}$.
Since $\{u,v\}\in E_{xy}$ and noting the edge consistency condition we have  $(\sigma(u)=\tilde{\sigma}(x)_{\pi_{x}(u)},\sigma(v)=\tilde{\sigma}(y)_{\pi_{y}(v)})\in {C}_{uv}$, as desired.
\end{description}

\paragraph{Runtime Analysis:} Finally, for a fixing of $\{x,y\}$ in $F$ and a fixing of $(a,b)\in\Sigma^d\times \Sigma^d$, the time needed to check if $(a,b)\in \tilde{C}_{xy}$ is at most $O(d^2\cdot |\Sigma|^2)$. Thus, the running time to produce $\Phi$ given $\Psi$ and $\Gamma$ is $O(|E|\cdot |\Sigma|^2+|F|\cdot |\Sigma|^{2d}\cdot d^2\cdot |\Sigma|^2)=O(|E|\cdot |F|\cdot |\Sigma|^{2d+2})$.
\end{proof}

An immediate corollary that will be useful to us is the following. 
\begin{corollary}[Putting together Theorem~\ref{thm:embedding} and Theorem~\ref{thm:embedCSP}]\label{cor:CSPembed}
There exists a constant $Z > 1$ and an algorithm $\mathcal{A}$ that takes as input  a 2-\CSP instance $\Gamma(G(V,E),\Sigma,\{C_{uv}\}_{\{u,v\}\in E})$ and an even integer $k\ge 6$,  and outputs a 2-\CSP instance $\Phi(H(X,F),\tilde\Sigma,\{\tilde{C}_{xy}\}_{\{x,y\}\in F})$ such that the following holds.
\begin{description}
\item[Size:] $H$ is a 3-regular bipartite graph with $|X|\le k$. We have $|\tilde\Sigma|=|\Sigma|^{Z\cdot \left(1+k^{-1}(|V|+|E|)\cdot \log k\right)}$.
    \item[Reduction Guarantee:]  $\Gamma$ is satisfiable if and only if $\Phi$ is satisfiable.
    \item[Runtime:] There is a polynomial $p:\N\to\N$ such that  $\mathcal{A}$ runs in time bounded above by $p\left(|E|\cdot k\cdot |\Sigma|^{\left(\frac{(|V|+|E|)\cdot \log k}{k}+1\right)}\right)$.
\end{description}
\end{corollary} 
\begin{proof}
Given as input a 2-\CSP instance $\Gamma(G(V,E),\Sigma,\{C_{uv}\}_{\{u,v\}\in E})$ and an integer $k$ the algorithm $\mathcal{A}$ of the corollary statement does the following. Let $\mathcal{A}_0$ and $\mathcal{A}_1$ 
 be the algorithms guaranteed through Theorem~\ref{thm:embedding} and Theorem~\ref{thm:embedCSP} respectively. 
 
 $\mathcal{A}$ first runs $\mathcal{A}_0$ on input $(G(V,E),k)$ and obtains as output 
a  bipartite 3-regular graph $H(X,F)$ with no isolated vertices  and a connected embedding $\Psi:V\to\P(X)$ such that   $|X|\le k$ and $\Delta(\Psi)\le  Z\cdot \left(1+k^{-1}(|V|+|E|)\right) \cdot \log k$. 

 \begin{sloppypar}$\mathcal{A}$ then runs $\mathcal{A}_1$ on input 
 $\left(\Gamma(G(V,E),\Sigma,\{C_{uv}\}_{\{u,v\}\in E}),H(X,F),\Psi: V\to \P(X)\right)$, to obtain a 2-\CSP instance $\Phi(H(X,F),\tilde\Sigma,\{\tilde{C}_{xy}\}_{\{x,y\}\in F})$ where $|\tilde\Sigma|$ is some integer bounded above by $ |\Sigma|^{Z\cdot \left(1+k^{-1}(|V|+|E|)\right)\cdot \log k}$. \end{sloppypar}
 
 Finally, $\mathcal{A}$ outputs $\Phi$. The reduction guarantee of the corollary statement follows from the reduction guarantee of Theorem~\ref{thm:embedCSP}. 

 The runtime of the algorithm is upper bounded by the sum of the runtimes of $\mathcal{A}_0$ and $\mathcal{A}_1$. The output of $\mathcal{A}_0$ is produced in time $(|V|+|E|)^{O(1)}$. 
 The output of $\mathcal{A}_1$ is produced in time $O(|E|\cdot |F|\cdot |\Sigma|^{2\cdot(1+Z\cdot \left(1+k^{-1}(|V|+|E|)\right)\cdot \log k)})$. The claim on the runtime in the corollary statement follows by noting that $|F|\le 1.5\cdot |X|\le 1.5\cdot k$.
\end{proof}
\subsection{\ETH Lower Bounds}

Our starting point for the lower bounds in this subsection is the following result on 3-coloring:

\begin{theorem}[\ETH Lower Bound for 3-coloring \cite{cygan2016tight}]\label{thm:3Col}
    Assuming \ETH, there is some $\delta>0$ such that no algorithm   can decide all 2-\CSP{} instances $\Gamma(G(V,E),\Sigma,\{C_{uv}\}_{\{u,v\}\in E})$, where $|\Sigma|=3$ and    $G$ is   4-regular, while running in $2^{\delta |V|}$ time. 
\end{theorem}

We remark that  in \cite{cygan2016tight} the lower bound is proven for 2-\CSP instances with max-degree 4, but this can be easily extended to 4-regular graphs by adding dummy constraints\footnote{A dummy constraint on two variables accepts every possible assignment to the two variables. } between pairs of vertices which are both of degree less than 4.


\main*
\begin{proof}
Suppose there is some $k\ge k_0$ and an algorithm $\mathcal{A}_k$ as suggested in this theorem statement for $\alpha:=\frac{\delta}{4Z}$ and $ k_0:= Z^2/\delta^2$, where $\delta$ is the constant  referred to in Theorem~\ref{thm:3Col} and $Z$ is the constant  referred to in Corollary~\ref{cor:CSPembed}. We will then design an algorithm $\tilde{\mathcal{A}}$ that runs in $2^{\nicefrac{\delta n}{2}}$ time and can decide 2-\CSP{}s on $n$ variables/vertices even when the alphabet set is of size 3 and the constraint graph is 4-regular. This would imply that \ETH is false (from Theorem~\ref{thm:3Col}). 

    Given $\tilde{\mathcal{A}}$ as input  a 2-\CSP instance $\Gamma(G(V,E),\Sigma,\{C_{uv}\}_{\{u,v\}\in E})$   where $|\Sigma|=3$ and 4-regular graph $G$, $\tilde{\mathcal{A}}$ first runs the algorithm in Corollary~\ref{cor:CSPembed} $(\Gamma,k)$ as input and obtains as output a 2-\CSP instance $\Phi(H(X,F),\tilde\Sigma,\{\tilde{C}_{xy}\}_{\{x,y\}\in F})$ with the guarantee that $|X|\le k$ and $|\tilde\Sigma|\le 3^{Z\cdot\left(1+\frac{|V|\log k}{k}\right)}$, for some universal constant $Z\ge  1$ (also note that the size of $\Phi$ is at most $O(k\cdot |\tilde\Sigma|)=O\left(k\cdot 3^{Z\cdot\left(1+\frac{|V|\log k}{k}\right)}\right)= O\left(k\cdot 3^{Z\cdot\left(1+\frac{|V|\log k_0}{k_0}\right)}\right)<2^{\delta |V|}$).

 Then $\tilde{\mathcal{A}}$ feeds $\Phi$ as input to $\mathcal{A}$  and obtains as output $\lambda\in\{0,1\}$. $\tilde{\mathcal{A}}$ outputs $\lambda$.  The reduction guarantee of the theorem statement follows from the reduction guarantee of Corollary~\ref{cor:CSPembed}. 

 The runtime of the algorithm is upper bounded by the sum of the runtimes of $\mathcal{A}_k$ and the algorithm in Corollary~\ref{cor:CSPembed}. The output of the algorithm in Corollary~\ref{cor:CSPembed} is produced in time $p\left(|V|\cdot k\cdot 3^{}\right)$ for some polynomial $p$. 
 The output of $\mathcal{A}$ is produced in time $|\tilde\Sigma|^{\alpha\cdot k/\log k}\le 3^{Z\cdot\left(1+\frac{|V|\log k}{k}\right)\cdot \alpha\cdot \frac{k}{\log k}}=3^{\alpha\cdot Z\cdot \frac{k}{\log k}+\alpha\cdot Z\cdot |V|}\le 3^{2\alpha\cdot Z\cdot |V|}\le 2^{4\alpha\cdot Z\cdot |V|}=2^{\delta \cdot |V|}$.  
\end{proof}

For most applications, the following corollary is sufficient. 
 
\begin{corollary}\label{cor:main} 
Assuming \ETH, there is no $f(k)\cdot |\Sigma|^{o(k/\log k)}$ time algorithm that can decide all  2-\CSP instances on 3-regular bipartite constraint graphs, where $\Sigma$ is the alphabet set, $k$ is the number of constraints, and $f$ is an arbitrary function. 
\end{corollary}

We now turn to prove Theorem~\ref{thm:mainwithupperbound}
.

\mainalso*
\begin{proof}
Suppose there is a function $f$ that is good and an algorithm $\mathcal{A}_f$ as suggested in this theorem statement for $\alpha:=\frac{\delta}{4Z}$ where $\delta$ is the constant  referred to in Theorem~\ref{thm:3Col} and $Z$ is the constant  referred to in Corollary~\ref{cor:CSPembed}. We will then design an algorithm $\tilde{\mathcal{A}}$ that runs in $2^{\nicefrac{\delta n}{2}}$ time and can decide 2-\CSP{}s on $n$ variables/vertices even when the alphabet set is of size 3 and the constraint graph is 4-regular. This would imply that \ETH is false (from Theorem~\ref{thm:3Col}). 

Given $\tilde{\mathcal{A}}$ as input  a 2-\CSP instance $\Gamma(G(V,E),\Sigma,\{C_{uv}\}_{\{u,v\}\in E})$   where $|\Sigma|=3$ and 4-regular graph $G$, $\tilde{\mathcal{A}}$ first computes an appropriate integer $k$ which it will later feed as input to the algorithm in Corollary~\ref{cor:CSPembed}.
Let $k_0$ be the smallest integer such that $f(k_0)\ge 2$. 
Let $k$ be the largest integer  such that $f(k-1)<3^{Z\cdot\left(1+\frac{|V|\log k}{k}\right)}$. Such  an integer $k$ can be computed in time $|V|^{O(1)}$ in the following way: We evaluate $f(1)$, $f(2)$, $\ldots$, and stop when either $f(k)\ge 3^{Z\cdot\left(1+\frac{|V|\log k}{k}\right)}$ or if computing $f(k)$ does not terminate after $|V|^2+k_0$ steps (in which case we can conclude that $f(k)>3^{Z\cdot\left(1+\frac{|V|\log k}{k}\right)}$).

    Now $\tilde{\mathcal{A}}$ first runs the algorithm in Corollary~\ref{cor:CSPembed} $(\Gamma,k)$ as input and obtains as output a 2-\CSP instance $\Phi(H(X,F),\tilde\Sigma,\{\tilde{C}_{xy}\}_{\{x,y\}\in F})$ with the guarantee that $|X|\le k$, $H(X,F)$ is 3-regular bipartite, and $|\tilde\Sigma|\le 3^{Z\cdot\left(1+\frac{|V|\log k}{k}\right)}\le f(k)$, for some universal constant $Z\ge  1$ (also note that the size of $\Phi$ is at most $O(k\cdot |\tilde\Sigma|)=O\left(k\cdot 3^{Z\cdot\left(1+\frac{|V|\log k}{k}\right)}\right) < 2^{\delta |V|}$).

 Then $\tilde{\mathcal{A}}$ feeds $\Phi$ as input to $\mathcal{A}$  and obtains as output $\lambda\in\{0,1\}$. $\tilde{\mathcal{A}}$ outputs $\lambda$.  The reduction guarantee of the theorem statement follows from the reduction guarantee of Corollary~\ref{cor:CSPembed}. 

 The runtime of the algorithm is upper bounded by the sum of the runtimes of $\mathcal{A}_f$ and the algorithm in Corollary~\ref{cor:CSPembed}. The output of the algorithm in Corollary~\ref{cor:CSPembed} is produced in time $p\left(|V|\cdot k\cdot 3^{}\right)$ for some polynomial $p$. 
 The output of $\mathcal{A}$ is produced in time $|\tilde\Sigma|^{\alpha\cdot k/\log k}\le 3^{Z\cdot\left(1+\frac{|V|\log k}{k}\right)\cdot \alpha\cdot \frac{k}{\log k}}=3^{\alpha\cdot Z\cdot \frac{k}{\log k}+\alpha\cdot Z\cdot |V|}\le 3^{2\alpha\cdot Z\cdot |V|}\le 2^{4\alpha\cdot Z\cdot |V|}=2^{\delta \cdot |V|}$.  
\end{proof}

The parameter setting of Theorem~\ref{thm:mainwithupperbound} may seem peculiar at first glance, but 
it has applications in lower bounds for classic (not parameterized) problems under ETH. 
For an illustration, consider a recent work on \textsc{Global Label MinCut}~\cite{DBLP:conf/soda/JaffkeLMPS23}.
The actual details of the considered problem are not important here; it suffices to say that the main
technical result of~\cite{DBLP:conf/soda/JaffkeLMPS23} is a parameterized reduction from 
2-\CSP{} to one parameterization of \textsc{Global Label MinCut} that,
given a 2-\CSP{} instance $\Gamma(G(V,E),\Sigma,\{C_{uv}\}_{\{u,v\}\in E})$ with $a := |V| + |E|$
produces an equivalent instance of \textsc{Global Label MinCut}
of size $2^{{O}(a \log a)} \cdot |\Sigma|$ with parameter value at most $a$.
This reduction, combined with Theorem~\ref{thm:mainwithupperbound}, refutes (assuming ETH) an existence of an algorithm
for \textsc{Global Label MinCut} with running time $n^{o(\log n/ (\log \log n)^2)}$, where $n$ is the size of the instance,
providing a nearly-tight lower bound to an existing quasipolynomial-time algorithm~\cite{DBLP:conf/soda/GhaffariKP17}.

Indeed, consider a good function $f(n) = n^n$.
Let $\Gamma(G(V,E),\Sigma,\{C_{uv}\}_{\{u,v\}\in E})$ be a 2-\CSP{} instance on $k$ variables, $G(V,E)$ 3-regular,
and alphabet set $\Sigma$ where $|\Sigma| < f(k)$.
The aforementioned reduction, applied to this 2-\CSP{} instance, produces a \textsc{Global Label MinCut}
instance of size bounded by 
\[ 2^{{O}(k \log k)} \cdot k^k = 2^{{O}(k \log k)}.\]
Now, a hypothetical algorithm solving this instance in time $n^{o(\log n / (\log \log n)^2)}$ solves
the original 2-\CSP{} instance in time $2^{o(k^2)}$, which is asymptotically smaller 
than $f(k-1)^{\alpha k / \log k}$ for a constant $\alpha > 0$.

\subsection*{Acknowledgement}
Karthik C.\ S.\ is supported by the National Science Foundation under Grant CCF-2313372 and by the Simons  
 Foundation, Grant Number 825876, Awardee Thu D. Nguyen. D\'{a}niel Marx is supported by the European Research Council (ERC) consolidator grant No.~725978 SYSTEMATICGRAPH. 
 During this research Marcin Pilipczuk was part of BARC, supported by the VILLUM Foundation grant 16582. 
 We
are also grateful to the  Dagstuhl Seminar 23291 for a special collaboration
opportunity.

\bibliographystyle{alpha}
\bibliography{references}
\end{document}